\documentclass[a4paper,UKenglish,cleveref, autoref, thm-restate]{lipics-v2021}
\usepackage{xspace}

\hideLIPIcs  

\title{Single-Exponential FPT Algorithms for Enumerating Secluded $\mathcal{F}$-Free Subgraphs and Deleting to Scattered Graph Classes}

\titlerunning{Single-Exponential FPT Algorithms for Enumerating Secluded $\mathcal{F}$-Free Subgraphs} 

\author{Bart M.\,P. Jansen 
}{Eindhoven University of Technology, The Netherlands}{b.m.p.jansen@tue.nl}{https://orcid.org/0000-0001-8204-1268}{}

\author{Jari J. H. {de Kroon}}{Eindhoven University of Technology, The Netherlands}{j.j.h.d.kroon@tue.nl}{https://orcid.org/0000-0003-3328-9712}{}

\author{Micha\l{} W\l{}odarczyk}{University of Warsaw, Poland}{m.wlodarczyk@mimuw.edu.pl}{https://orcid.org/0000-0003-0968-8414}{}

\authorrunning{B.M.P. Jansen, J.J.H. de Kroon, and M. W\l{}odarczyk}

\funding{\flag[3cm]{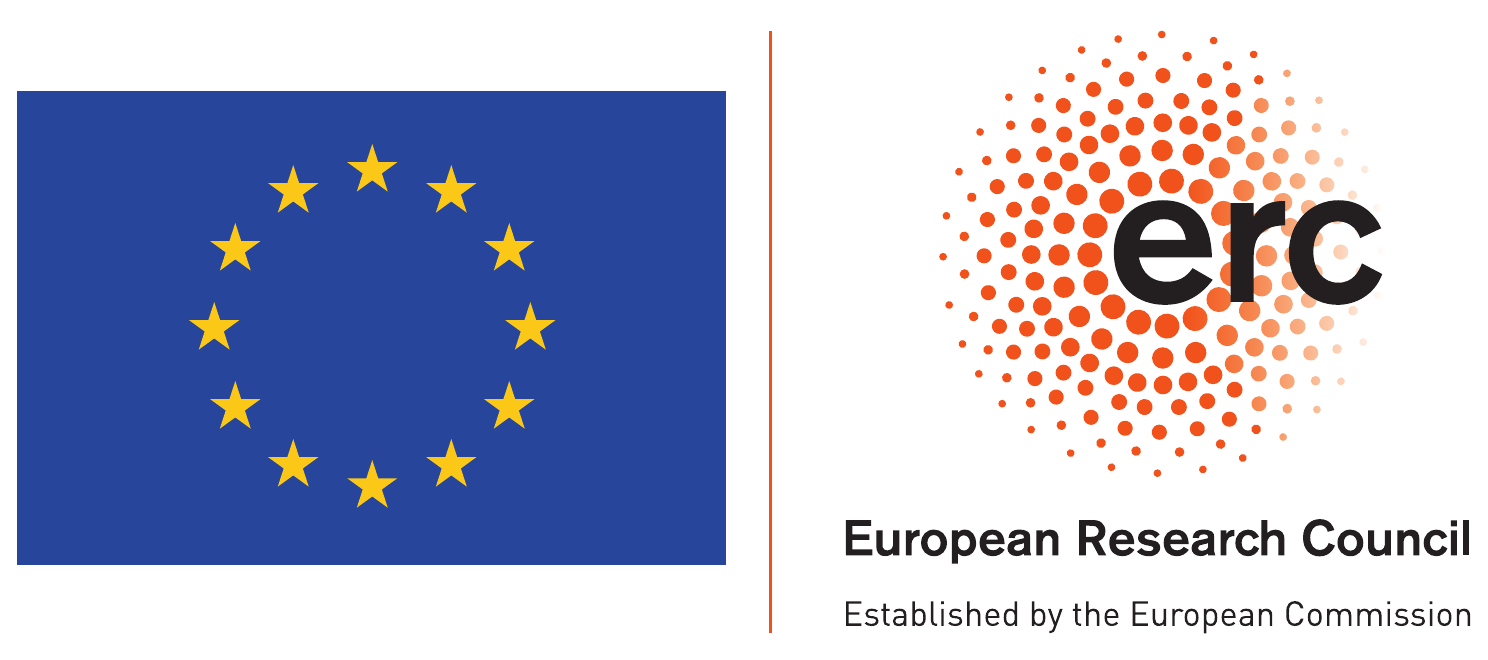}This project has received funding from the European Research Council (ERC) under the European Union's Horizon 2020 research and innovation programme (grant agreement No 803421, ReduceSearch).} 

\Copyright{Bart M.\,P. Jansen, Jari J. H. de Kroon, and Micha\l{} W\l{}odarczyk}



\category{} 

\relatedversion{} 

\nolinenumbers 

\EventEditors{Satoru Iwata and Naonori Kakimura}
\EventNoEds{2}
\EventLongTitle{34th International Symposium on Algorithms and Computation (ISAAC 2023)}
\EventShortTitle{ISAAC 2023}
\EventAcronym{ISAAC}
\EventYear{2023}
\EventDate{December 3-6, 2023}
\EventLocation{Kyoto, Japan}
\EventLogo{}
\SeriesVolume{283}
\ArticleNo{38}

\newcommand{\defparproblem}[4]{
	\vspace{2mm}
	\noindent\fbox{
		\begin{minipage}{0.96\linewidth}
			\begin{tabular*}{\linewidth}{@{\extracolsep{\fill}}lr} \textsc{#1} & {\bf{Parameter:}} #3 \\ \end{tabular*}
			{\bf{Input:}} #2 \\
			{\bf{Task:}} #4
		\end{minipage}
	}
	\vspace{2mm}
}

\newcommand{\defparquestion}[4]{
	\vspace{2mm}
	\noindent\fbox{
		\begin{minipage}{0.96\linewidth}
			\begin{tabular*}{\linewidth}{@{\extracolsep{\fill}}lr} \textsc{#1} & {\bf{Parameter:}} #3 \\ \end{tabular*}
			{\bf{Input:}} #2 \\
			{\bf{Question:}} #4
		\end{minipage}
	}
	\vspace{2mm}
}

\newcommand{\kopt}{$k$-locally irredundant\xspace}

\newcommand{\lambdastar}{\lambda^{\mathrm{L}}}

\newcommand{\Oh}{\mathcal{O}}
\newcommand{\hh}{\ensuremath{\mathcal{H}}}
\newcommand{\ff}{\ensuremath{\mathcal{F}}}

\newcommand{\sm}{\setminus}

\newcommand{\hhdel}{\textsc{$\hh$-deletion}\xspace}

\ifdefined\DEBUG{}
\newcommand{\mic}[1]{{\color{blue}{#1}}}

\def\rem#1{{\marginpar{\raggedright\scriptsize #1}}}
\newcommand{\micr}[1]{\rem{\textcolor{blue}{\(\bullet \) #1}}}

\newcommand{\bmp}[1]{{\color{purple}{#1}}}
\newcommand{\bmpr}[1]{\rem{\textcolor{purple}{\(\bullet \) #1}}}

\newcommand{\jjh}[1]{{\color{orange}{#1}}}
\newcommand{\jjhr}[1]{\rem{\textcolor{orange}{\(\bullet \) #1}}}

\else
\newcommand{\mic}[1]{#1}
\newcommand{\bmp}[1]{#1}
\newcommand{\jjh}[1]{#1}
\newcommand{\micr}[1]{ }
\newcommand{\bmpr}[1]{ }
\newcommand{\jjhr}[1]{ }
\fi

\def\SHORTVERSION{true}
\ifdefined\SHORTVERSION{}
\newcommand{\shortv}[1]{#1}
\newcommand{\longv}[1]{}
\else
\newcommand{\shortv}[1]{}
\newcommand{\longv}[1]{#1}
\fi

\begin{document}

\maketitle

\begin{abstract}
The celebrated notion of important separators bounds the number of small $(S,T)$-separators in a~graph which are `farthest from~$S$' in a technical sense.  
In this paper, we introduce a generalization of this powerful algorithmic primitive, tailored to undirected graphs, that is phrased in terms of \emph{$k$-secluded} vertex sets: sets with an open neighborhood of size at most~$k$.

In this terminology, the bound on 
important separators says that there are at most~$4^k$ maximal $k$-secluded connected vertex sets~$C$ containing~$S$ but disjoint from~$T$. We generalize this statement significantly:
even when we demand that~$G[C]$ avoids a finite set~$\mathcal{F}$ of forbidden induced subgraphs, the number of such maximal subgraphs is~$2^{\mathcal{O}(k)}$ and they can be enumerated efficiently. This enumeration algorithm allows us to make significant improvements for two problems from the~literature.

Our first application concerns the \textsc{Connected $k$-Secluded $\mathcal{F}$-free subgraph} problem, where $\mathcal{F}$ is a finite set of forbidden induced subgraphs. Given a graph in which each vertex has a~positive integer weight, the problem asks to find a maximum-weight connected $k$-secluded vertex set~$C \subseteq V(G)$ such that~$G[C]$ does not contain an induced subgraph isomorphic to any~$F \in \mathcal{F}$. The parameterization by~$k$ is known to be solvable in triple-exponential time via the technique of recursive understanding, which we improve to single-exponential.

Our second application concerns the deletion problem to \emph{scattered graph classes}. A scattered graph class is defined by demanding that every connected component is contained in at least one of the prescribed graph classes~$\Pi_1, \ldots, \Pi_d$. The deletion problem to a scattered graph class is to find a vertex set of size at most~$k$ whose removal yields a graph from the class. We obtain a~single-exponential algorithm whenever each class~$\Pi_i$ is characterized by a finite number of forbidden induced subgraphs. This generalizes and improves upon earlier results in the literature.

\end{abstract}
\clearpage

\section{Introduction} \label{sec:intro}


Graph separations have played a central role in algorithmics since the discovery of min-cut/max-flow duality and the polynomial-time algorithm to compute a maximum flow~\cite{FordF56}. Nowadays, more complex separation properties are crucial in the study of parameterized complexity, where the goal is to design algorithms for NP-hard problems whose running time can be bounded as~$f(k) \cdot n^{\Oh(1)}$ for some function~$f$ that depends only on the \emph{parameter~$k$} of the input. There are numerous graph problems which either explicitly involve finding separations of a certain kind (such as \textsc{Multiway Cut}~\cite{Marx06}, \textsc{Multicut}~\cite{BousquetDT18,MarxR14}, \textsc{$k$-Way Cut}~\cite{KawarabayashiT11}, and \textsc{Minimum Bisection}~\cite{CyganLPPS19}) or in which separation techniques turn out to be instrumental for an efficient solution (such as \textsc{Directed Feedback Vertex Set}~\cite{ChenLLOR08} and \textsc{Almost 2-SAT}~\cite{RazgonO09}).

The field of parameterized complexity has developed a robust toolbox of techniques based on graph separators, e.g., treewidth reduction~\cite{MarxOR13}, important separators~\cite{Marx11}, shadow removal~\cite{MarxR14}, discrete relaxations~\cite{CyganPPW13,Guillemot11a,IwataWY16,IwataYY18},  protrusion replacement~\cite{Misra16a}, randomized contractions and recursive understanding~\cite{ChitnisCHPP16,CyganKLPPSW21,LokshtanovR0Z18}, and flow augmentation~\cite{KimKPW21,KimKPW22}. These powerful techniques allowed a large variety of graph separation problems to be classified as fixed-parameter tractable. However, this power comes at a cost. The running times for many applications of these techniques are superexponential: of the form~$2^{p(k)} \cdot n^{\Oh(1)}$ for a high-degree polynomial~$p$, double-exponential, or even worse. Discrete relaxations form a notable exception, which we discuss in Section~\ref{sec:conclusion}.

The \mic{new} algorithmic \jjh{primitive} we develop can be seen \jjh{as an extension of} important separators~\cite{Marx11}~\cite[\S 8]{CyganFKLMPPS15}. The study of important separators was pioneered by Marx~\cite{Marx06,Marx11} and refined by follow-up work by several authors~\cite{ChenLL09,LokshtanovM13}, which was recognized by the EATCS-IPEC Nerode Prize 2020~\cite{BodlaenderWW20}. The technique is used to bound the number of extremal $(S,T)$-separators in an $n$-vertex graph~$G$ with vertex sets~$S$ and~$T$. The main idea is that, even though the number of distinct inclusion-minimal $(S,T)$-separators (which are vertex sets potentially intersecting~$S \cup T$) of size at most~$k$ can be as large as~$n^{\Omega(k)}$, the number of \emph{important} separators which leave a maximal vertex set reachable from~$S$, is bounded by~$4^k$. 
For \textsc{Multiway Cut}, a pushing lemma~\cite[Lem.~6]{Marx06} shows that there is always an optimal solution that contains an important separator, which leads to an algorithm solving the problem in time~$2^{\Oh(k)} \cdot n^{\Oh(1)}$. Important separators also form a key ingredient for solving many other problems such as \textsc{Multicut}~\cite{BousquetDT18,MarxR14} and \textsc{Directed Feedback Vertex Set}~\cite{ChenLLOR08}.

For our purposes, it will be convenient to view the bound on the number of important separators through the lens of \emph{secluded subgraphs}.

\mic{
\begin{definition} \label{def:seclusion:maximal}
A vertex set~$S \subseteq V(G)$ or induced subgraph~$G[S]$ of an undirected graph~$G$ is said to be {\em $k$-secluded} if~$|N_G(S)| \leq k$, that is, the number of vertices outside~$S$ which are adjacent to a vertex of~$S$ is bounded by~$k$.

{A vertex set~$S$ in a graph~$G$ is called \emph{seclusion-maximal} with respect to a~certain property~$\Pi$ if~$S$ satisfies~$\Pi$ and for all sets~$S' \supsetneq S$ that satisfy~$\Pi$ we have~$|N_G(S')| > |N_G(S)|$.}
\end{definition}
}

Hence a~seclusion-maximal set with property $\Pi$ is inclusion-maximal among all subsets with the same size neighborhood. Consequently, the number of inclusion-maximal $k$-secluded sets satisfying 
$\Pi$ is at most the number of seclusion-maximal $k$-secluded sets with that~property.
 
Using the terminology of seclusion-maximal subgraphs, the bound on the number of important~$(S,T)$-separators of size at most~$k$ in a graph~$G$ is equivalent to the following statement: in~the~graph~$G'$ obtained from~$G$ by inserting a new source~$r$ adjacent to~$S$, the number of \emph{seclusion-maximal} $k$-secluded connected subgraphs~$C$ containing~$r$ but no vertex of~$T$ is bounded~by~$4^k$. The neighborhoods of such subgraphs~$C$ correspond exactly to the important~$(S,T)$-separators~in~$G$.

While a number of previously studied cut problems~\cite{LokshtanovPSSZ20,MarxOR13} place further restrictions on the vertex set that forms the separator (for example, requiring it to induce a connected graph or independent set) our generalization instead targets the structure of the $k$-secluded connected subgraph~$C$. We will show that, for any fixed finite family~$\mathcal{F}$ of graphs, the number of $k$-secluded connected subgraphs~$C$ as above which are seclusion-maximal with respect to satisfying the additional requirement that~$G[C]$ contains no induced subgraph isomorphic to a member of~$\mathcal{F}$ is still bounded by~$2^{\Oh(k)}$. Observe that the case~$\mathcal{F} = \emptyset$ corresponds to the original setting of important separators. Note that a priori, it is not even clear that the number of seclusion-maximal graphs of this form can be bounded by any function~$f(k)$, let alone a single-exponential~one. 

\paragraph*{Our contribution}
Having introduced the background of secluded subgraphs, we continue by stating our result exactly. This will be followed \jjh{by} a discussion on its applications.

For a finite set~$\mathcal{F}$ of graphs we define~$||\mathcal{F}|| := \max_{F \in \mathcal{F}} |V(F)|$, the maximum order of any graph in~$\mathcal{F}$. We say that a graph is~$\mathcal{F}$-free if it does not contain an \emph{induced} subgraph isomorphic to a~graph in~$\mathcal{F}$. Our generalization of important separators is captured by the following theorem, in which we use~$\Oh_{\mathcal{F}}(\ldots)$ to indicate that the hidden constant depends~on~$\mathcal{F}$.

\begin{restatable}{theorem}{thmSecludedFFreeEnum}\label{thm:secluded-F-free-enum}
Let~$\mathcal{F}$ be a finite set of 
graphs. For any $n$-vertex graph~$G$, non-empty vertex set~$S \subseteq V(G)$, potentially \bmp{empty~$T \subseteq V(G) \setminus S$}, and integer~$k$, the number of $k$-secluded \bmp{induced subgraphs~$G[C]$ which are} seclusion-maximal \bmp{with respect to being} connected, $\mathcal{F}$-free, and satisfying~$S \subseteq C \subseteq V(G) \setminus T$, is bounded by~$2^{\Oh_{\ff}(k)}$. A superset of size $2^{\Oh_{\ff}(k)}$ of these subgraphs can be enumerated in time~$2^{\Oh_{\ff}(k)} \cdot n^{||\mathcal{F}||+\Oh(1)}$ and polynomial space.
\end{restatable}

The single-exponential bound given by the theorem is best-possible in several ways. Existing lower bounds on the number of important separators~\cite[Fig. 8.5]{CyganFKLMPPS15} imply that even when~$\mathcal{F} = \emptyset$ the bound cannot be improved to~$2^{o(k)}$. The term~$n^{||\mathcal{F}||}$ in the running time is unlikely to be avoidable, 
since even testing whether a single graph is~$\mathcal{F}$-free is equivalent to \textsc{Induced Subgraph Isomorphism} and cannot be done in time~$n^{o(||\mathcal{F}||)}$~\cite[Thm. 14.21]{CyganFKLMPPS15} assuming the Exponential Time Hypothesis (ETH) due to lower bounds for \textsc{$k$-Clique}.

The polynomial space bound applies to the internal space usage of the algorithm, as the output size may be exponential in~$k$. More precisely, we consider polynomial-space algorithms equipped with a command that outputs an element and we require that for each element in the enumerated set, this command is called at least once. 
\mic{The algorithm could also enumerate just the set in question (rather than its superset) by postprocessing the output and comparing each pair of enumerated subgraphs. However, storing the entire output requires exponential space.}

{By executing the enumeration algorithm for every singleton set $S$ of the form $\{v\}$, $v \in V(G)$, and $T = \emptyset$, we immediately obtain the following.}

\begin{corollary}\label{cor:secluded-F-free-enum-all}
Let~$\mathcal{F}$ be a finite set of 
graphs. For any $n$-vertex graph~$G$ and integer~$k$, the number of $k$-secluded {induced subgraphs~$G[C]$ which are} seclusion-maximal {with respect to being} connected and $\mathcal{F}$-free is~$2^{\Oh_{\ff}(k)} \cdot n$. A superset of size $2^{\Oh_{\ff}(k)} \cdot n$ of these subgraphs can be enumerated in time~$2^{\Oh_{\ff}(k)} \cdot n^{||\mathcal{F}||+\Oh(1)}$ and polynomial space.
\end{corollary}

{Note that we require that the set $\ff$ of forbidden induced subgraphs is finite. This is}
necessary in order to obtain a bound independent of~$n$ in \cref{thm:secluded-F-free-enum}. For example, the number of seclusion-maximal $(k=1)$-secluded connected subgraphs~$C$ containing a prescribed vertex~$r$ for which~$C$ induces an acyclic graph is already as large as~$n-1$ in a graph consisting of a single cycle, since each way of omitting a vertex other than~$r$ gives such a subgraph. For this case, the forbidden induced subgraph characterization~$\mathcal{F}$ consists of all cycles. Extending this example to a flower structure of~$k$ cycles of length~$n/k$ pairwise intersecting only in~$r$ shows that the number of seclusion-maximal $k$-secluded $\mathcal{F}$-free connected subgraphs containing~$r$ is~$\Omega(n^k / k^k)$ and cannot be bounded by~$f(k) \cdot n^{\Oh(1)}$ for any function~$f$.

{We give two applications of~\cref{thm:secluded-F-free-enum} to improve the running time of existing super-exponential (or even triple-exponential) parameterized algorithms to single-exponential, which is optimal under ETH. For each application, we start by presenting some context.}

\paragraph*{Application I: Optimization over connected $k$-secluded $\mathcal{F}$-free subgraphs} The computation of secluded versions of graph-theoretic objects such as paths~\cite{BevernFT20,ChechikJPP17,LuckowF20}, trees~\cite{DonkersJK22}, Steiner trees~\cite{FominGKK17}, or feedback vertex sets~\cite{BevernFMMSS18}, has attracted significant attention over recent years. This task becomes hard already for detecting $k$-secluded disconnected sets satisfying very simple properties.
In particular, detecting \jjh{a} $k$-secluded independent set of size~$s$ is W[1]-hard when parameterized by~$k+s$~\cite{BevernFMMSS18}. 

Golovach, Heggernes, Lima, and Montealegre~\cite{GolovachHLM20} suggested then to focus on \emph{connected} $k$-secluded subgraphs and studied the problem of finding one, which belongs to a graph class $\hh$, of maximum total weight. They therefore studied the \textsc{Connected $k$-secluded $\mathcal{F}$-free subgraph} problem for a finite family~$\mathcal{F}$ of forbidden induced subgraphs. Given an undirected graph~$G$ in which each vertex~$v$ has a positive integer weight~$w(v)$, and an integer~$k$, the problem is to find a maximum-weight connected $k$-secluded vertex set~$C$ for which~$G[C]$ is $\mathcal{F}$-free. They presented an algorithm based on recursive understanding to solve the problem in time~$2^{2^{2^{\Oh_\mathcal{F}(k \log k)}}} \cdot n^{\Oh_\mathcal{F}(1)}$. 
{We improve the dependency on $k$ to single-exponential.}
\begin{restatable}{corollary}{corConnectedFFree}\label{cor:connected:ffree}
For each fixed finite family~$\mathcal{F}$, \textsc{Connected $k$-secluded $\mathcal{F}$-free subgraph} can be solved in time~$2^{\Oh_{\mathcal{F}}(k)} \cdot n^{||\mathcal{F}|| + \Oh(1)}$ and polynomial space.
\end{restatable}
{This result follows 
directly from 
\cref{cor:secluded-F-free-enum-all}} since a maximum-weight $k$-secluded $\mathcal{F}$-free subgraph must be seclusion-maximal. Hence it suffices to check for each enumerated subgraph whether it is~$\mathcal{F}$-free, and remember the heaviest one for which this is the case.

{The parameter dependence of our algorithm for \textsc{Connected $k$-secluded $\mathcal{F}$-free subgraph} is optimal under ETH. This follows from an easy reduction from \textsc{Maximum Independent Set}, which cannot be solved in time~$2^{o(n)}$ under ETH~\cite[Thm. 14.6]{CyganFKLMPPS15}. Finding a maximum independent set in an $n$-vertex graph~$G$ is equivalent to finding a maximum-weight triangle-free {connected} induced $(k=n)$-secluded subgraph in the graph~$G'$ that is obtained from~$G$ by inserting a universal vertex of weight~$n$ and setting the weights of all other vertices to~$1$. Consequently, an algorithm with running time~$2^{o(k)} \cdot n^{\Oh(1)}$ for \textsc{Connected $k$-secluded triangle-free induced subgraph} would violate ETH and our parameter dependence is already optimal for~$\mathcal{F} = \{K_3\}$.}

\paragraph*{Application II: Deletion to scattered graph classes} When there are several distinct graph classes (e.g., split graphs and claw-free graphs) on which a problem of interest (e.g. \textsc{Vertex Cover}) becomes tractable, it becomes relevant to compute a minimum vertex set whose removal ensures that each resulting component belongs to one such tractable class. This can lead to fixed-parameter tractable algorithms for solving the original problem on inputs which are \emph{close} to such so-called \emph{islands of tractability}~\cite{GanianRS17b}. The corresponding optimization problem has been coined the deletion problem to \emph{scattered} graph classes~\cite{JacobKMR23,JacobM023}. Jacob, Majumdar, and Raman~\cite{DBLP:conf/iwpec/JacobM020} (later joined by de Kroon for the journal version~\cite{JacobKMR23}) consider the \textsc{$(\Pi_1,\dots,\Pi_d)$-deletion} problem; given hereditary graph classes $\Pi_1,\dots,\Pi_d$, find a set $X\subseteq V(G)$ of at most $k$ vertices such that each connected component of $G-X$ belongs to $\Pi_i$ for some $i \in [d]$. Here $d$ is seen as a constant. When the set of forbidden induced subgraphs $\mathcal{F}_i$ of $\Pi_i$ is finite for each $i \in [d]$, they show~\cite[Lem.~12]{JacobKMR23} that the problem is solvable in time $2^{q(k)+1} \cdot n^{\Oh_\Pi(1)}$, where~$q(k) = 4k^{{10(pd)^2}+4} + 1$. Here $p$ is the maximum number of vertices of any forbidden induced subgraph.

Using Theorem~\ref{thm:secluded-F-free-enum} as a black box, we obtain a single-exponential algorithm for this problem.

\begin{restatable}{theorem}{thmInducedScattered} \label{thm:induced:scattered}
\textsc{$(\Pi_1,\dots,\Pi_d)$-deletion} can be solved in time $2^{\Oh_\Pi(k)}\cdot n^{\Oh_\Pi(1)}$ and polynomial space when each graph class~$\Pi_i$ is characterized by a finite set~$\mathcal{F}_i$ of (not necessarily connected) forbidden induced subgraphs.
\end{restatable}

The main idea behind the algorithm is the following. For an arbitrary vertex~$v$, either it belongs to the solution, or we may assume that in the graph that results by removing the solution, the vertex~$v$ belongs to a connected component that forms a seclusion-maximal connected $k$-secluded $\mathcal{F}_i$-free induced subgraph of~$G$ for some~$i \in [d]$. Branching on each of the~$2^{\Oh_\Pi(k)}$ options gives the desired running time by exploiting the fact that in most recursive calls, the parameter decreases by more than a constant (cf.~\cite[Thm.~8.19]{CyganFKLMPPS15}). Prior to our work, single-exponential algorithms were only known for a handful of ad-hoc cases where~$d=2$, such as deleting to a graph in which each component is a tree or a clique~\cite{JacobKMR23}, or when one of the sets of forbidden induced subgraphs~$\mathcal{F}_i$ contains a path.

{Similarly as our first application, the resulting algorithm for \textsc{$(\Pi_1,\dots,\Pi_d)$-deletion} is ETH-tight: the problem is a strict generalization of \textsc{$k$-Vertex Cover}, which is known not to admit an algorithm with running time~$2^{o(k)} \cdot n^{\Oh(1)}$ unless ETH fails.}

\paragraph*{Techniques}
The proof of \cref{thm:secluded-F-free-enum} is based on a bounded-depth search tree algorithm with a nontrivial progress measure.
By adding vertices to~$S$ or~$T$ in branching steps of the enumeration algorithm, the sets grow and the size of a minimum $(S,T)$-separator increases accordingly. The size of a minimum $(S,T)$-separator disjoint from~$S$ is an important progress measure for the algorithm: if it ever exceeds~$k$, there can be no $k$-secluded set containing all of~$S$ and none of~$T$ and therefore the enumeration is finished. 

The branching steps are informed by the farthest minimum~$(S,T)$-separator (see Lemma~\ref{lem:closestfarthestsep}), similarly as the enumeration algorithm for important separators, but are significantly more involved \bmp{because we have to handle the forbidden induced subgraphs.}
A distinctive feature of our algorithm is that the decision made by branching can be to add certain vertices to the set~$T$, while the important-separator enumeration only branches by enriching~$S$. A key step is to use submodularity to infer that a certain vertex set is contained in \emph{all} seclusion-maximal secluded subgraphs under consideration when other branching steps are inapplicable. 

As an illustrative example consider the case $\ff = \{K_3\}$, that is, we want to enumerate seclusion-maximal vertex sets $C \subseteq V(G) \sm T$, $C \supseteq S$, which induce connected triangle-free subgraphs with at most $k$ neighbors.
Let $\lambdastar(S,T)$ denote the size of a minimum vertex set disjoint from $S$ that separates $T$ from $S$---\mic{we will refer to such separators as {\em left-restricted}.
Then $\lambdastar(S,T)$ corresponds to the minimum possible size of $N(C)$.}
Similarly to the enumeration algorithm for important separators, we keep track of two measures: (M1) the value of $k$, and (M2) the gap between $k$ and $\lambdastar(S,T)$.
We combine them into a single progress measure which is bounded by $2k$ and decreases \mic{during branching}.

The first branching scenario occurs when there is some triangle in the graph $G$ which intersects or is adjacent to $S$; 
then we guess which of its vertices should belong to $N(C)$, remove it from the graph, and decrease $k$ by one.
Otherwise, let $\mathcal{U} = \{U_1, \dots U_d\}$ be the collection of all vertex sets of triangles in $G$ (which are now disjoint from $S$).
When there exists a triangle $U_i$ whose addition to $T$ increases the value $\lambdastar(S,T)$, we branch into two possibilities: either $U_i$ is disjoint from $N[C]$---then we set $T \leftarrow T \cup U_i$ so the measure (M2) decreases---or $U_i$ intersects $N(C)$---then we perform branching as above.
\mic{We show that in the remaining case
all the triangles are separated from $S$ by the minimum left-restricted $(S,T)$-separator closest to $S$; hence
the value of $\lambdastar(S,T)$ equals the value of $\lambdastar(S,T \cup V(\mathcal{U}))$.} Next, let $P$ be the farthest minimum \mic{left-restricted} $(S,T \cup V(\mathcal{U}))$-separator; we use submodularity to justify that we can now safely add to $S$ all the vertices reachable from $S$ in $G-P$.
This allows us to assume that when $u \in P$ then either $u \in N(C)$ or $u \in C$, \jjh{which} leads to the last branching strategy.
We either delete $u$ \mic{(so $k$ drops)} or 
add $u$ to $S$; note that in this~case the progress measure may not change directly.
The key observation is that adding $u$ to $S$ invalidates the farthest $(S,T\cup V(\mathcal{U}))$-separator $P$ and now we are promised to make progress in the very next branching step. 
The different branching scenarios are illustrated in~Figure~\ref{fig:ffree}.

\begin{figure}
    \centering
    \includegraphics[width=\linewidth]{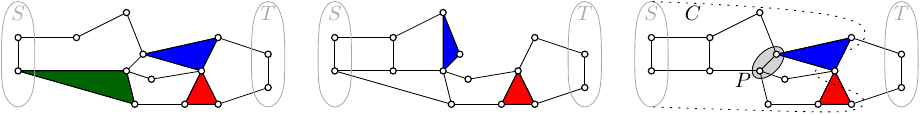}
    \caption{Illustration of the branching steps for enumerating triangle-free $k$-secluded subgraphs for~$k=3$. Left: the green triangle intersects~$S$; we branch to guess which vertex belongs to~$N(C)$. Middle: setting where~$2=\lambdastar(S,T) < \lambdastar(S,T \cup V(\mathcal{U})) = 3$; adding the top triangle to~$T$ increases~$\lambdastar$. The set $\mathcal{U}$ consists of the colored triangles. Right: setting where~$\lambdastar(S,T) = \lambdastar(S,T \cup V(\mathcal{U})) = 2$, with a corresponding farthest separator~$P$. In this case every seclusion-maximal triangle-free set $C \supseteq S$ must be a superset of the reachability set of $S$ in $G-P$.}
    \label{fig:ffree}
\end{figure}

The only property of $K_3$ that we have relied on is connectivity: if a triangle intersects a~triangle-free set~$C$ then it must intersect~$N(C)$ as well.
This is no longer true when $\ff$ contains a~disconnected graph.
For example, the forbidden family for the class of split graphs includes~$2K_2$.
A subgraph of $F \in \ff$ that can be obtained by removing some  components from $F$ is called a \emph{partial forbidden graph}. 
We introduce a third measure to keep track of how many different partial forbidden graphs appear as induced subgraph in $G[S]$.
The main difficulty in 
generalizing the previous approach lies in justification of the greedy argument:  
when $P$ is a farthest minimum separator between $S$ and a certain set then we want to replace $S$ with the set $S'$ of vertices reachable from $S$ in $G-P$.
In the setting of connected obstacles this fact could be proven easily because $S'$ was disjoint from all the obstacles.
The problem is now it may contain some partial forbidden subgraphs.
We handle this issue by defining $P$ in such a way that the sets of partial forbidden graphs appearing in $G[S]$ and $G[S']$ are the same and
giving a \jjh{rearrangement} argument about subgraph isomorphisms. 
This allows us to extend the analysis to any family $\ff$ of forbidden subgraphs.

\paragraph*{Organization} 

The remainder of the paper is organized as follows. 
We provide formal preliminaries in Section~\ref{sec:prelims}. 
The algorithm for enumerating secluded $\mathcal{F}$-free subgraphs is presented in Section~\ref{sec:secludedffree}.
{Then in \cref{sec:applications} we apply it to improve the running times of the two discussed problems.} We conclude in Section~\ref{sec:conclusion}.

\section{Preliminaries} \label{sec:prelims}

\shortv{\paragraph*{Graphs and separators}}
\longv{\subsection{Graphs and separators}}
We consider finite, simple, undirected graphs. We denote the vertex and edge sets of a graph~$G$ by $V(G)$ and $E(G)$ respectively, with $|V(G)| = n$ and $|E(G)| = m$. For a set of vertices $S \subseteq V(G)$, by $G[S]$ we denote the graph induced by $S$. We use shorthand $G-v$ and $G-S$ for $G[V(G) \setminus \{v\}]$ and $G[V(G) \setminus S]$, respectively. The open neighborhood $N_G(v)$ of $v \in V(G)$ is defined as $\{u \in V(G) \mid \{u,v\} \in E(G)\}$. The closed neighborhood of~$v$ is $N_G[v] = N_G(v) \cup \{v\}$. For $S \subseteq V(G)$, we have $N_G[S] = \bigcup_{v \in S} N_G[v]$ and $N_G(S) = N_G[S] \setminus S$.
The set $C$ is called connected if the graph $G[C]$ is connected.

We proceed by introducing notions concerning separators which are crucial for the branching steps of our algorithms. For two sets $S,T \subseteq V(G)$ in a graph~$G$, a set $P \subseteq V(G)$ is an unrestricted $(S,T)$-separator if no connected component of $G-P$ contains a vertex from both $S \setminus P$ and $T \setminus P$. Note that such a separator may intersect $S \cup T$. Equivalently,~$P$ is an~$(S,T)$-separator if each $(S,T)$-path contains a vertex of $P$. 
A restricted~$(S,T)$-separator is an unrestricted $(S,T)$-separator $P$ which satisfies~$P \cap (S \cup T) = \emptyset$.
{A left-restricted~$(S,T)$-separator is an unrestricted $(S,T)$-separator $P$ which satisfies~$P \cap S = \emptyset$.} Let $\lambdastar_G(S,T)$ denote the minimum size of a left-restricted~$(S,T)$-separator, or $+\infty$ if no such separator exists (which happens when~$S \cap T \neq \emptyset$).

\begin{theorem}[Ford-Fulkerson]\label{thm:prelim:ford}
There is an algorithm that, given an $n$-vertex $m$-edge graph $G = (V,E)$, disjoint sets $S,T \subseteq V(G)$, and an integer $k$, runs in time $\Oh(k(n+m))$ and determines whether there exists a restricted $(S,T)$-separator of size at most $k$.
If so, then the algorithm returns a separator of minimum size.
\end{theorem} 

By the following observation we can translate properties of restricted separators into properties of left-restricted separators.

\begin{observation}\label{obs:prelim:restricted-translate} 
Let $G$ be a graph and  $S,T \subseteq V(G)$.
Consider the graph $G'$ obtained from $G$ by adding a new vertex $t$ adjacent to each $v \in T$.
Then $P \subseteq V(G)$ is a left-restricted $(S,T)$-separator in $G$ if and only if $P$ is a restricted $(S,t)$-separator in $G'$.
\end{observation}

\shortv{\paragraph*{Extremal separators and submodularity}}
\longv{\subsection{Extremal separators and submodularity}}
{The following submodularity property of the cardinality of the open neighborhood is well-known; cf.~\cite[\S 44.12]{Schrijver03} and~\cite[Fn.~3]{KratschW20}.}

\begin{lemma}[Submodularity]\label{lem:submodularity}
Let $G$ be a graph and $A, B \subseteq V(G)$. Then the following holds:
$$|N_G(A)| + |N_G(B)| \ge |N_G(A \cap B)| + |N_G(A \cup B)|.$$
\end{lemma}

For a graph~$G$ and vertex sets~$S, P \subseteq V(G)$, we denote by $R_G(S, P)$ the set of vertices which can be reached in~$G-P$ from at least one vertex in the set~$S \setminus P$.


\begin{lemma}
\label{lem:closestfarthestsep}
Let $G$ be a graph and $S,T \subseteq V(G)$ be two disjoint non-adjacent vertex sets. There exist minimum restricted $(S,T)$-separators $P^-$ (closest) and $P^+$ (farthest), such that for each minimum restricted $(S,T)$-separator $P$, it holds that $R_G(S,P^-) \subseteq R_G(S,P) \subseteq R_G(S,P^+)$. Moreover, if a minimum restricted $(S,T)$-separator has size $k$, then $P^-$ and $P^+$ can be identified in $\Oh(k(n+m))$ time.
\end{lemma}
{
\begin{proof}
{It is well-known (cf.~\cite[Thm.~8.5]{CyganFKLMPPS15} for the edge-based variant of this statement, or~\cite[\S 3.2]{KratschW20} for the same concept with slightly different terminology) that the existence of these separators follows from submodularity (\cref{lem:submodularity}), 
while they can be computed by analyzing the residual network when applying the Ford-Fulkerson algorithm to compute a minimum separator. We sketch the main ideas for completeness. 

By merging~$S$ into a single vertex~$s^+$ and merging~$T$ into a single vertex~$t^-$, which is harmless because a restricted separator is disjoint from~$S \cup T$, we may assume that~$S$ and~$T$ are singletons. Transform~$G$ into an edge-capacitated directed flow network $D$ in which~$s^+$ is the source and~$t^-$ is the sink. All remaining vertices~$v \in V(G) \setminus (S \cup T)$ are split into two representatives~$v^-, v^+$ connected by an arc~$(v^-,v^+)$ of capacity~$1$. 
For each edge~$uv \in E(G)$ with $u,v \in V(G) \setminus \{s^+,t^-\}$ we add arcs~$(u^+, v^-), (u^-, v^+)$ of capacity~$2$. For edges of the form $s^+v$ we add an arc $(s^+,v^-)$ of capacity 2 to $D$. Similarly, for edges of the form $t^-v$ we add an arc $(v^+,t^-)$ of capacity 2. 
Then the minimum size~$k$ of a restricted~$(S,T)$-separator in~$G$ equals the maximum flow value in the constructed network, which can be computed by~$k$ rounds of the Ford-Fulkerson algorithm. Each round can be implemented to run in time~$\Oh(n+m)$. From the state of the residual network when Ford-Fulkerson terminates we can extract~$P^-$ and~$P^+$ as follows: the set~$P^-$ contains all vertices~$v \in V(G) \setminus (S \cup T)$ for which the source can reach~$v^-$ but not~$v^+$ in the final residual network. Similarly,~$P^+$ contains all vertices~$v \in V(G) \setminus (S \cup T)$ for which~$v^+$ can reach the sink but~$v^-$ cannot.}
\end{proof}
}

By Observation~\ref{obs:prelim:restricted-translate}, we can apply the lemma above for left-restricted separators too; when the sets $S,T$ are disjoint, then $S$ is non-adjacent to $t$ in the graph obtained by adding a vertex $t$ adjacent to every vertex in $T$.

The extremal separators identified in Lemma~\ref{lem:closestfarthestsep} explain when adding a vertex to~$S$ or~$T$ increases the separator size.
{The following statement is not symmetric because we work with the non-symmetric notion of a left-restricted separator.}

\begin{lemma}
\label{lemma:separator:increase}
Let~$G$ be a graph, let~$S,T$ be disjoint vertex sets, and let~$P^-$ and~$P^+$ be the closest and farthest minimum \emph{left-restricted}~$(S,T)$-separators. Then for any vertex~$v \in V(G)$, the following holds:
\begin{enumerate}
    \item \label{increase:type1} $\lambdastar_G(S \cup \{v\}, T) > \lambdastar_G(S,T)$ if and only if~$v \in R_G(T,P^+) \cup P^+$. 
    \item \label{increase:type2} $\lambdastar_G(S, T \cup \{v\}) > \lambdastar_G(S,T)$ if and only if~$v \in R_G(S,P^-)$.
\end{enumerate}
\end{lemma}
{
\begin{proof}
Adding a vertex to $S$ or~$T$ can never decrease the separator size, so for both cases, the left-hand side is either equal to or strictly greater than the right-hand side. 

\subparagraph*{\eqref{increase:type1}.} Observe that if $v \notin R_G(T,P^+) \cup P^+$, then $P^+$ is also a left-restricted~$(S \cup \{v\},T)$-separator which implies $\lambdastar_G(S \cup \{v\}, T) = \lambdastar_G(S,T)$. If $v \in T$, then \eqref{increase:type1} holds as~$\lambdastar_G(S \cup \{v\},T) = +\infty$.
Consider now $v \in (R_G(T,P^+) \cup P^+) \setminus T$; we argue that adding it to~$S$ increases the separator size. Assume for a contradiction that there exists a \jjh{minimum} left-restricted~$(S \cup \{v\},T)$-separator $P$ of size at most $\lambdastar_G(S,T) = \jjh{|P^+|}$.
\jjh{Note that since $P$ is left-restricted, we have $v \notin P$. Observe that $P$ is also a left-restricted $(S,T)$-separator. By Lemma~\ref{lem:closestfarthestsep} we have $R_G(S,P) \subseteq R_G(S,P^+)$. Since $v \in (R_G(T,P^+) \cup P^+) \setminus T$, it follows that $v \notin R_G(S,P)$. We do a case distinction on $v$ to construct a path $Q$ from $v$ to $T$.

\begin{itemize}
    \item In the case that $v \in P^+ \setminus T$, then since $P^+$ is a minimum separator it must be inclusion-minimal. Therefore, since $P^+ \setminus \{v\}$ is not an $(S,T)$-separator, it follows that $v$ has a neighbor in $R_G(T,P^+)$ and so there is a path $Q$ from $v$ to $T$ in the graph induced by $R_G(T,P^+) \cup \{v\}$ such that $V(Q) \cap P^+ = \{v\}$.
    \item In the case that $v \in R_G(T,P^+) \setminus T$, then by definition there is a path from $v$ to $T$ in the graph induced by $R_G(T,P^+)$.
\end{itemize}

Since $P$ is a left-restricted $(S \cup \{v\},T)$-separator and therefore $v \notin P$, it follows that $P$ contains at least one vertex $u \in V(Q)$ that is not in $R_G(S,P^+) \cup P^+$. Let $P'$ be the set of vertices adjacent to $R_G(S,P)$. Since all vertices of $P'$ belong to $P$ while $u \notin P'$, it follows that $P'$ is a left-restricted $(S,T)$-separator that is strictly smaller than $P$, a contradiction to $|P| \leq \lambdastar_G(S,T)$. }

\subparagraph*{\eqref{increase:type2}.} If $v \notin R_G(S, P^-)$, then $P^-$ is a left-restricted~$(S,T\cup \{v\})$-separator as well which implies $\lambdastar_G(S, T\cup \{v\}) = \lambdastar_G(S,T)$.
If $v \in R_G(S, P^-)$, suppose that there exists a minimum left-restricted~$(S, T\cup \{v\})$-separator $P$ of size $|P^-|$.
\jjh{Note that $v \notin S$, as otherwise no such separator exists. Furthermore $P$ is also a left-restricted $(S,T)$-separator. By~\cref{lem:closestfarthestsep} we have $R_G(S,P^-) \subseteq R_G(S,P)$. But since $v \notin R_G(S,P)$ we reach a contradiction as $R_G(S,P) \not \supseteq R_G(S,P^-)$.}
%
\end{proof}
}

The following lemma captures the idea that if~$\lambdastar_G(S,\jjh{T \cup Z}) > \lambdastar_G(S, T)$, then there is a single vertex from~$Z$ whose addition to~$T$ already increases the size of a minimum left-restricted~$(S,T)$-separator. We will use it  to argue that when it is cheaper to separate~$S$ from~$T$ than to separate~$S$ from~$T$ together with all obstacles of a certain form, then there is already a single vertex from one such obstacle which causes this increase.

\begin{lemma}
\label{lem:separatingZUTisoT}
Let $G$ be a graph,~$S \subseteq V(G)$, and~$T,Z \subseteq V(G) \setminus S$. If there is no vertex $v \in Z$ such that $\lambdastar_G(S,T \cup \{v\}) > \lambdastar_G(S,T)$, then $\lambdastar_G(S,T) = \lambdastar_G(S,T \cup Z)$. Furthermore if $\lambdastar_G(S,T) \leq k$, then in $\Oh(k(n+m))$ time we can either find such a vertex $v$ or {determine that} no such vertex exists.
\end{lemma}
{
\begin{proof}
Let $P^-$ be the minimum left-restricted $(S,T)$-separator which is closest to~$S$.
If for every $v \in Z$ the value of $\lambdastar_G(S,T \cup \{v\})$ equals $\lambdastar_G(S,T)$ then \cref{lemma:separator:increase} implies that each $v \in Z$ lies outside $R_G(S,P^-)$ so $Z \cap R_G(S,P^-) = \emptyset$.
Then $P^-$ is a left-restricted $(S,T \cup Z)$-separator of size $\lambdastar_G(S,T)$.

On the other hand, if there is a vertex $v \in Z$ for which $\lambdastar_G(S,T \cup \{v\}) > \lambdastar_G(S,T)$ then $v \in R_G(S,P^-)$.
Hence, in order to detect such a vertex it suffices to compute the closest minimum left-restricted $(S,T)$-separator $P^-$, which can be done in time $\Oh(k(n+m))$ via \cref{lem:closestfarthestsep}.
\end{proof}
}

Finally, the last lemma of this section uses submodularity to argue that the neighborhood size of a vertex set~$C$ with~$S \subseteq C \subseteq V(G) \setminus T$ does not increase when taking its union with the reachable set~$R_G(S,P)$ with respect to a minimum left-restricted $(S,T)$-separator~$P$.

\begin{lemma}
\label{lem:neighborhoodsize:union}
If~$P \subseteq V(G)$ is a minimum left-restricted~$(S,T)$-separator in a graph~$G$ and~$S' = R_G (S,P)$, then for any set~$C$ with~$S\subseteq C \subseteq V(G) \setminus T$ we have~$|N_G(C \cup S')| \leq |N_G(C)|$.
\end{lemma}
{
\begin{proof}
Observe that since~$P$ is a minimum left-restricted~$(S,T)$-separator, we have~$|P| = \lambdastar_G(S,T)$ and~$P = N_G(S')$. We apply the submodular inequality to the sets~$C$ and~$S'$.
$$|N_G(C)|+|N_G(S')| \geq |N_G(C \cup S')| + |N_G(C \cap S')| \geq |N_G(C \cup S')| + \lambdastar_G(S,T).$$
Here the last step comes from the fact that~$S \subseteq S' \subseteq V(G) \setminus T$ since it is the set reachable from~$S$ with respect to a left-restricted~$(S,T)$-separator, so that~$C \cap S'$ contains all of~$S$ and is disjoint from~$T$. This implies that~$N_G (C \cap S')$ is a left-restricted~$(S,T)$-separator, so that~$|N_G(C \cap S')| \geq \lambdastar_G(S,T)$. 

As~$|N_G(S')| = |P| =\lambdastar_G(S,T)$, canceling these terms from both sides \jjh{gives}~$|N_G(C)| \geq |N_G(C \cup S')|$ which completes the proof.
\end{proof}
} 

\section{The enumeration algorithm}
\label{sec:secludedffree}

We need the following concept to deal with forbidden subgraphs which may be disconnected.

\begin{definition}
A partial forbidden graph~$F'$ is a graph obtained from some~$F \in \mathcal{F}$ by deleting zero or more connected components. (So each $F \in \mathcal{F}$ itself is also considered a partial forbidden graph.)
\end{definition}

We use the following notation to work with induced subgraph isomorphisms. An induced subgraph isomorphism from~$H$ to~$G$ is an injection~$\phi \colon V(H) \to V(G)$ such that for all distinct~$u,v \in V(H)$ we have~$\{u,v\} \in E(H)$ if and only if $\{\phi(u), \phi(v)\} \in E(G)$. For a vertex set~$U \subseteq V(H)$ we let~$\phi(U) := \{\phi(u) \mid u \in U\}$. For a subgraph~$H'$ of~$H$ we write~$\phi(H')$ instead of~$\phi(V(H'))$.

The following definition will be important to capture the progress of the recursive algorithm. {See Figure~\ref{fig:squares} for an illustration.}

\begin{definition}
We say that a vertex set~$U \subseteq V(G)$ enriches a vertex set~$S \subseteq V(G)$ with respect to~$\mathcal{F}$ if there exists a partial forbidden graph~$F'$ such that~$G[S \cup U]$ contains an induced subgraph isomorphic to~$F'$ but~$G[S]$ does not. We call such a set $U$ an enrichment.

An enrichment~$U$ is called tight if~$U = \phi(F') \setminus S$ for some induced subgraph isomorphism $\phi \colon V(F') \to V(G)$ from some partial forbidden graph~$F'$ for which~$G[S]$ does not contain an induced subgraph isomorphic to~$F'$.
\end{definition}

The following observation will be used to argue for the correctness of the 
\mic{recursive scheme.
Note that we get an implication only in one way (being seclusion-maximal in $G$ implies being  seclusion-maximal in $G-v$, not the other way around), which is the reason why we output a superset of the sought set in \cref{thm:secluded-F-free-enum}.} 


\begin{observation}\label{obs:seclusion:maximal:deleteneighbor}
Let $G$ be a graph containing disjoint sets~$S,T \subseteq V(G)$ and let $C \subseteq V(G)$ be seclusion-maximal with respect to being connected, $\ff$-free, $k$-secluded and satisfying~$S \subseteq C \subseteq V(G) \setminus T$. For each $v \in N_G(C)$ it holds that $C$ is seclusion-maximal in~$G-v$ with respect to being connected, $\ff$-free, $(k-1)$-secluded and satisfying~$S \subseteq C \subseteq V(G - v) \setminus T$.
\end{observation}

With these ingredients, we present the enumeration algorithm.
Recall that~$||\mathcal{F}|| = \max_{F \in \mathcal{F}} |V(F)|$ denotes the maximum order of any graph in~$\mathcal{F}$.

\begin{figure}
    \centering
    \includegraphics[width=.9\linewidth]{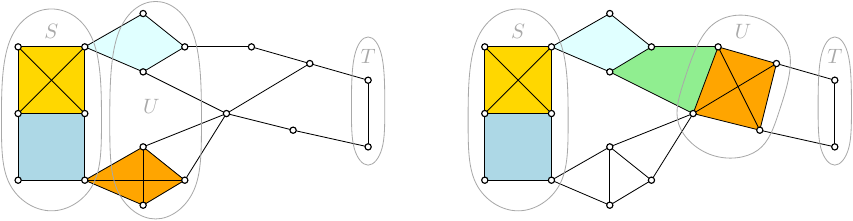}
    \caption{Illustration of the idea of enrichment and the branching steps in the proof of Theorem~\ref{thm:secluded-F-free-enum}. Here $F = C_4 \uplus K_4$.
    Left: The graph $G[S]$ contains $C_4$ and $K_4$, but not $F$. The set~$U$ enriches~$S$ since $G[S \cup U]$ contains a new partial forbidden graph $F$. Every component of $G[U]$ is adjacent to $S$, so Step~\ref{stepf:local} applies.
    Right: The two top copies of $C_4$ do not enrich $S$. One of them intersects the only copy of $K_4$ in $G[S]$; the other one is adjacent to the only copy of~$K_4$, while~$F$ has to appear as an induced subgraph. However the connected set $U$ enriches $S$ and it gets detected in Step~\ref{stepf:u}.
    In~both cases the enrichments are tight.
    }
    \label{fig:squares}
\end{figure}

\thmSecludedFFreeEnum*

\begin{proof}
Algorithm $\mathsf{Enum}_{\mathcal{F}}(G,S,T,\jjh{k})$ solves {the} enumeration task as follows.
\begin{enumerate}
    \item \label{stepf:stop} Stop the algorithm if one of the following holds:
    \begin{enumerate}
        \item\label{stepf:large_flow} $\lambdastar_G(S,T) > k$,
        \item the vertices of~$S$ are not contained in a single connected component of~$G$, or
        \item the graph~$G[S]$ contains an induced subgraph isomorphic to some~$F \in \mathcal{F}$.
    \end{enumerate}
    \emph{There are no secluded subgraphs satisfying all imposed conditions.}
    \item\label{stepf:output} If the connected component~$C$ of~$G$ which contains~$S$ is $\mathcal{F}$-free and includes no vertex of~$T$: output~$C$ and stop.
    
    \emph{Component~$C$ is the unique seclusion-maximal one satisfying the imposed conditions.}
    \item \label{stepf:local} If there is a vertex set~$U \subseteq V(G) \setminus (S \cup T)$ 
    such that:
        \begin{itemize}
            \item each connected component of~$G[U]$ is adjacent to a vertex of~$S$, and
            \item the set~$U$ is a tight enrichment of~$S$ with respect to~$\mathcal{F}$ (so~$G[S \cup U]$ contains a new partial forbidden graph)
        \end{itemize}
        
    	then execute the following calls and stop:
    	\begin{enumerate}
    	    \item\label{stepf:local:branch} For each~$u \in U$ call $\mathsf{Enum}_{\mathcal{F}}(G-u,S,T,k-1)$.
    	    \item\label{stepf:local:all} Call $\mathsf{Enum}_{\mathcal{F}}(G,S \cup U,T, k)$.
    	\end{enumerate}
    	\emph{A tight enrichment can have at most $||\ff||$ vertices which bounds the branching factor in Step~\ref{stepf:local:branch}.
    	Note that these are exhaustive even though we do not consider adding~$U$ to~$T$: since each component of~$G[U]$ is adjacent to a vertex of~$S$, if a relevant secluded subgraph does not contain all of~$U$ then it contains some vertex of~$U$ in its neighborhood and we find it in Step~\ref{stepf:local:branch}.}
	\item \label{stepf:u} For the rest of the algorithm, let~$\mathcal{U}$ denote the collection of all connected vertex sets~$U \subseteq V(G) \setminus (S \cup T)$ 
	which form tight enrichments of $S$ with respect to~$\mathcal{F}$. Let~$V(\mathcal{U}) := \bigcup_{U \in \mathcal{U}} U$. 
	\begin{enumerate}
	    \item\label{stepf:u:a} If~$\lambdastar_G(S,T)<\lambdastar_G(S,T \cup V(\mathcal{U}))$: then (using Lemma~\ref{lem:separatingZUTisoT}) there exists~$U\in \mathcal{U}$ such that $\lambdastar_G(S,T \cup U) > \lambdastar_G(S,T)$, execute the following calls and stop:
	    \begin{enumerate}
	        \item\label{stepf:u:a:i} For each~$u \in U$ call $\mathsf{Enum}_{\mathcal{F}}(G-u,S,T, k-1)$. (The value of $k$ decreases.)
	        \item\label{stepf:u:a:ii} Call $\mathsf{Enum}_{\mathcal{F}}(G,S \cup U,T, k)$. (We absorb a new partial forbidden graph.)
			\item\label{stepf:u:a:iii} Call $\mathsf{Enum}_{\mathcal{F}}(G,S,T \cup U, k)$. (The separator size increases.)
	    \end{enumerate}
		\item \label{stepf:push} If~$\lambdastar_G (S,T)= \lambdastar_G(S,T \cup V(\mathcal{U}))$, then let~$P$ be the farthest left-restricted minimum~$(S,T \cup V(\mathcal{U}))$-separator in~$G$, and let~$S' = R_G(S,P) \supseteq S$. Pick an arbitrary~$p \in P$ (which may be contained in~$T$ but not in~$S$).
		\begin{enumerate}
		    \item \label{stepf:push:neighbor} Call $\mathsf{Enum}_{\mathcal{F}}(G-p,S',T \setminus \{p\}, k-1$). (The value of $k$ decreases.)
		    \item \label{stepf:delayprogress} If~$p \notin T$, then call $\mathsf{Enum}_{\mathcal{F}}(G,S' \cup \{p\},T, k)$. \\ (Either here or in the next iteration we will be able to make progress.)		    
		\end{enumerate}
    \end{enumerate}

    \emph{It might happen that $\mathcal{U}$ is empty; in this case the algorithm will execute Step \ref{stepf:push}.
    Also note that $P$ is non-empty because the algorithm did not stop in Step~\ref{stepf:output}; hence it is always possible to choose a~vertex~$p \in P$.
    }

\end{enumerate}

Before providing an in-depth analysis of the algorithm, we establish that it always terminates. For each recursive call, either a vertex outside $S$ is deleted, or one of $S$ or $T$ grows in size while the two remain disjoint. Since $S$ and $T$ are vertex subsets of a finite graph, this process terminates.
The key argument in the correctness of the algorithm is formalized in the following claim.

\begin{claim}\label{claim:secluded:push:correct}
If the algorithm reaches Step~\ref{stepf:push}, then every seclusion-maximal $k$-secluded subgraph satisfying the conditions of the theorem statement contains~$S'$.
\end{claim}
\begin{claimproof}
We prove the claim by showing that for an arbitrary $k$-secluded $\mathcal{F}$-free connected induced subgraph~$G[C]$ satisfying~$S \subseteq C \subseteq V(G) \setminus T$, the subgraph induced by~$C \cup S'$ also satisfies these properties while~$|N_G(C \cup S')| \leq |N_G(C)|$. Hence any seclusion-maximal subgraph satisfying the conditions contains~$S'$.

Under the conditions of Step~\ref{stepf:push}, we have~$\lambdastar_G (S,T)= \lambdastar_G(S,T \cup V(\mathcal{U}))$, so that the set~$P$ is a left-restricted minimum~$(S,T)$-separator. 
\mic{Next, we have $S' = R_G(S,P)$.
By exploiting submodularity of the size of the open neighborhood, we prove in \cref{lem:neighborhoodsize:union} 
that $|N_G(C \cup S')| \leq |N_G(C)|$.} The key part of the argument is to prove that~$C \cup S'$ induces an $\mathcal{F}$-free subgraph. Assume for a contradiction that~$G[C \cup S']$ contains an induced subgraph isomorphic to~$F \in \mathcal{F}$ and let~$\phi \colon V(F) \to C \cup S'$ denote an induced subgraph isomorphism. Out of all ways to choose~$\phi$, fix a choice that minimizes the number of vertices~$|\phi(F) \setminus S|$ the subgraph uses from outside~$S$. We distinguish two cases.

\paragraph*{{Neighborhood of $S$ intersects $\phi(F)$}} If~$\phi(F) \cap N_G(S) \neq \emptyset$, then we will use the assumption that Step~\ref{stepf:local} of the algorithm was not applicable to derive a contradiction. Let~$F'$ be the graph consisting of those connected components~$F_i$ of~$F$ for which~$\phi(F_i) \cap N_G[S] \neq \emptyset$; 
let~$U = \phi(F') \sm S$.
Observe that each connected component of~$G[U]$ is adjacent to a vertex of~$S$. 
    By construction~$U$ is disjoint from~$S$, \jjh{and}~$U$ is disjoint from~$T$ since~$\phi(F) \subseteq C \cup S'$ while both these sets are disjoint from~$T$.
    Hence~$U$ satisfies all but one of the conditions for applying Step~\ref{stepf:local}. Since the algorithm reached Step~\ref{stepf:push}, it follows that~$U$ failed the last criterion which means that 
    the partial forbidden graph~$F'$ 
    also exists as an induced subgraph in~$G[S]$. Let~$\phi_{F'} \colon V(F') \to S$ be an induced subgraph isomorphism from~$F'$ to~$G[S]$. Since all vertices~$v \in V(F)$ for which~$\phi(v) \in N_G[S]$ satisfy~{$v \in V(F')$}, we can define a new subgraph isomorphism~$\phi'$ of~$F$ in~$G[C \cup S']$ as follows for each~$v \in V(F)$:
    \begin{equation} \label{eq:phiprime}
        \phi'(v) = \begin{cases}
        \phi_{F'}(v) & \mbox{if~$v \in F'$} \\
        \phi(v) & \mbox{otherwise.}
        \end{cases}
    \end{equation}
    Observe that this is a valid induced subgraph isomorphism since~$F'$ consists of some connected components of~$F$, and we effectively replace the model of~$F'$ by~$\phi_{F'}$. Since the model of the remaining graph~$\overline{F'} = F - F'$ does not use any vertex of~$N_G[S]$ by definition of~$F'$, there are no edges between vertices of~$\phi_{F'}(F')$ and vertices of~$\phi(\overline{F'})$, which validates the induced subgraph isomorphism.
    
    Since~$\phi(F)$ contains at least one vertex from~$N_G(S)$ while~$\phi'(F)$ does not, and the only vertices of~$\phi'(F) \setminus \phi(F)$ belong to~$S$, we conclude that~$\phi'(F)$ contains strictly fewer vertices outside~$S$ than~$\phi(F)$; a contradiction to minimality of~$\phi$.

    \paragraph*{{Neighborhood of $S$ does not intersect $\phi(F)$}} Now suppose that~$\phi(F) \cap N_G(S) = \emptyset$. If~$\phi(F) \subseteq C$, then~$\phi(F)$ is an induced $F$-subgraph in~$G[C]$, a contradiction to the assumption that~$C$ is $\mathcal{F}$-free. Hence~$\phi(F)$ must contain a vertex~$v \in S' \setminus C \subseteq S' \setminus S$. Since the previous case was not applicable,~$v \notin N_G(S)$ and therefore~$v \in S' \setminus N_G[S]$.

    Fix an arbitrary connected component~$F_i$ of $F$ for which $\phi(F_i)$ contains a vertex of~$S' \setminus N_G[S]$. We derive several properties of~$\phi(F_i)$.
    \begin{enumerate}
        \item Since~$F_i$ is a connected component of~$F$, the graph~$G[\phi(F_i)]$ is connected.
        \item We claim that $\phi(F_i) \cap S = \emptyset$. Note that a connected subgraph cannot both contain a vertex from~$S$ and a vertex outside~$N_G[S]$ without intersecting~$N_G(S)$. Since~$\phi(F) \cap N_G(S) = \emptyset$ by the case distinction, the graph~$G[\phi(F_i)]$ is connected since~$F_i$ is connected, and~$\phi(F_i)$ contains a vertex of~$S' \setminus N_G[S]$, we find $\phi(F_i) \cap S = \emptyset$.
        \item $\phi(F_i) \cap T = \emptyset$, since~$\phi(F) \subseteq C \cup S'$ while both~$C$ and~$S'$ are disjoint from~$T$.
        \item We claim that~$\phi(F_i) \notin \mathcal{U}$. To see that, recall that~$S' = R_G(S,P)$ is the set of vertices reachable from~$S$ when removing the~$(S,T \cup V(\mathcal{U}))$-separator~$P$. The definition of separator therefore ensures that no vertex of~$S'$ belongs to~$V(\mathcal{U})$. Since~$\phi(F_i)$ contains a vertex of~$S' \setminus N_G[S]$ by construction, some vertex of~$\phi(F_i)$ does not belong to~$V(\mathcal{U})$ and therefore~$\phi(F_i) \notin \mathcal{U}$.
    \end{enumerate}
    Now note that~$\phi(F_i)$ satisfies almost all requirements for being contained in the set~$\mathcal{U}$ defined in Step~\ref{stepf:u}: 
    it induces a connected subgraph and it is disjoint from~$S \cup T$. From the fact that~$\phi(F_i) \notin \mathcal{U}$ we therefore conclude that it fails the last criterion: the set~$\phi(F_i)$ is not a tight enrichment of $S$.
    
    Let~$F'$ be the graph formed by~$F_i$ together with all components~$F_j$ of~$F$ for which~$\phi(F_j) \subseteq S$; then $\phi(F_i) = \phi(F') \sm S$.
    Since~$\phi(F_i)$ is not a tight enrichment of $S$, the partial forbidden graph~$F'$ 
    is also contained in~$G[S]$. Let~$\phi_{F'} \colon F' \to S$ denote an induced subgraph isomorphism of~$F'$ to~$G[S]$. Since~$\phi(F)$ contains no vertex of~$N_G(S)$, we can define a new subgraph isomorphism~$\phi'$ of~$F$ in~$G[C \cup S']$ exactly as in~\eqref{eq:phiprime}. 
    
    Since the graph~$F'$ consists of some connected components of~$F$, while~$\phi_{F'}(F') \subseteq S$ and~$\phi(\overline{F'}) \cap N_G[S] = \emptyset$, it follows that~$\phi'$ is an induced subgraph isomorphism of~$F$ in~$G[C \cup S']$. 
    But~$|\phi'(F) \setminus S|$ is strictly smaller than~$|\phi(F) \setminus S|$ since~$\phi(F_i)$ intersects~$S' \setminus N_G[S]$ while~$\phi'(F_i) \subseteq \phi'(F') \subseteq S$ \jjh{and}~$\phi$ and~$\phi'$ coincide on~$\overline{F'}$.
    This contradicts the minimality of the choice of~$\phi$.

Since the case distinction is exhaustive, this proves the claim.
\end{claimproof}

Using the previous claim, we can establish the correctness of the algorithm.
\begin{claim} \label{claim:secluded:correct}
If~$G[C]$ is an induced subgraph of~$G$ that is seclusion-maximal with respect to being connected, $\mathcal{F}$-free, $k$-secluded and satisfying~$S \subseteq C \subseteq V(G) \setminus T$, then~$C$ occurs in the output of $\mathsf{Enum}_{\mathcal{F}}(G,S,T,k)$.
\end{claim}
\begin{claimproof}
We prove this claim by induction on the recursion depth of the $\mathsf{Enum}_\ff$ algorithm, which is valid as we argued above it is finite. In the base case, the algorithm does not recurse. In other words, the algorithm either stopped in Step~\ref{stepf:stop} or~\ref{stepf:output}. If the algorithm stops in Step~\ref{stepf:stop}, then there can be no induced subgraph satisfying the conditions and so there is nothing to show. If the algorithm stops in Step~\ref{stepf:output}, then the only seclusion-maximal induced subgraph is the $\ff$-free connected component containing $S$. Note that this component is $k$-secluded since $k \geq 0$ as $\lambdastar_G(S,T) \geq 0$ and the algorithm did not stop in Step~\ref{stepf:large_flow}.

For the induction step, we may assume that each recursive call made by the algorithm correctly enumerates a superset of the seclusion-maximal subgraphs satisfying the conditions imposed by the parameters of the recursive call, as the recursion depth of the execution of those calls is strictly smaller than the recursion depth for the current arguments~$(G,S,T,k)$. 
Consider a connected $\ff$-free $k$-secluded induced subgraph $G[C]$ of $G$ with $S \subseteq C \subseteq V(G) \setminus T$  that is seclusion-maximal with respect to satisfying all these conditions.
Suppose there is a vertex set $U \subseteq V(G) \setminus (S \cup T)$ 
that satisfies the conditions of Step~\ref{stepf:local}. If $U \subseteq C$, then by induction $C$ is part of the enumerated output of Step~\ref{stepf:local:all}. Otherwise, since each connected component of $G[U]$ is adjacent to a vertex in $S$, there is at least one vertex $u \in U$ such that $u \in N_G(C)$. By Observation~\ref{obs:seclusion:maximal:deleteneighbor}, the output of the corresponding call in Step~\ref{stepf:local:branch} contains $C$. Note that since $U \cap T = \emptyset$, we have $T \subseteq V(G) \setminus (S \cup U)$ and therefore the recursive calls satisfy the input requirements.

Next we consider the correctness in case such a set $U$ does not exist so the algorithm reaches Step~\ref{stepf:u}.
Let $\mathcal{U}$ be the set of tight enrichments as defined in Step~\ref{stepf:u}. First suppose that $\lambdastar_G(S,T)<\lambdastar_G(S,T \cup V(\mathcal{U}))$. Then by the contrapositive of the first part of Lemma~\ref{lem:separatingZUTisoT} with $Z = V(\mathcal{U})$, there is a vertex $v \in V(\mathcal{U}) \setminus T$ such that $\lambdastar_G(S,T \cup \{v\}) > \lambdastar_G(S,T)$. By picking an enrichment $U \in \mathcal{U}$ such that $v \in U$, this implies $\lambdastar_G(S,T \cup U) > \lambdastar_G(S,T)$. Now if there is a vertex $u \in U$ such that $u \in N_G(C)$, then by induction and Observation~\ref{obs:seclusion:maximal:deleteneighbor} we get that $C$ is output by the corresponding call in Step~\ref{stepf:u:a:i}. Otherwise, either $U \subseteq C$ or $U \cap C = \emptyset$ (since $U$ is connected) and $C$ is found in Step~\ref{stepf:u:a:ii} or Step~\ref{stepf:u:a:iii} respectively. Again observe that these recursive calls satisfy the input requirements as $U \cap (S \cup T) = \emptyset$.

Finally suppose that $\lambdastar_G(S,T) = \lambdastar_G(S,T \cup V(\mathcal{U}))$. By Claim~\ref{claim:secluded:push:correct} we get that $S' \subseteq C$. We first argue that $P = N_G(S')$ is non-empty. Note that since the algorithm did not stop in Step~\ref{stepf:stop}, the graph $G[S]$ is $\ff$-free and $S$ is contained in a single connected component of $G$. Furthermore since it did not stop in Step~\ref{stepf:output}, the connected component containing $S$ either has a vertex of $T$ or is not $\ff$-free. Note that the former case already implies $\lambdastar_G(S,T) > 0$. If the component has no vertex of $T$ and is not $\ff$-free, then it contains a vertex set~$J$ for which~$G[J]$ is isomorphic to some $F \in \ff$. Observe that $J \setminus (S \cup T) = J \setminus S$ is a tight enrichment of $S$. 
We have established that it is possible to enrich $S$ but we need an enrichment that meets the conditions of Step~\ref{stepf:u}.
Let $U \subseteq V(G) \sm (S \cup T)$ be a tight enrichment of minimum size and let $\phi \colon V(F') \to V(G)$ be the corresponding subgraph isomorphism from some partial forbidden graph $F'$; we have $U = \phi(F') \sm S$.
We argue that $G[U]$ is connected.
If each connected component of $G[U]$ is adjacent to a vertex of $S$, then Step~\ref{stepf:local} would have applied, contradicting the fact that the algorithm reaches Step~\ref{stepf:u}.
Hence, there exists a connected component of $G[U]$ that is non-adjacent to $S$; let~$U'$ be the vertex set of such a component.
Since $U$ is chosen to be minimum, we get that $U \setminus U'$ is not a tight enrichment, and so there is an induced subgraph of $G[S]$ isomorphic to the partial forbidden graph {$F'' = G[\phi(F') \sm U']$}.
This subgraph of~$G[S]$ combines with the graph~$G[U']$ to form an induced subgraph isomorphic to~$F'$ (we exploit that~$U'$ is not adjacent to~$S$),  which shows that~$U'$ is a tight enrichment.
By minimality of $U$ we obtain $U = U'$. Hence~$U$ is not adjacent to~$S$ and the graph $G[U]$ is connected so $U \in \mathcal{U}$.
Since $U$ and $S$ are contained in the same connected component we get that $\lambdastar_G(S,T \cup V(\mathcal{U})) > 0$. This implies there exists some vertex $p \in P = N_G(S')$. Since $S' \subseteq C$, we either get $p \in N_G(C)$, or (if $p \notin T$) $p \in C$. By induction (and Observation~\ref{obs:seclusion:maximal:deleteneighbor}) we conclude that $C$ is part of the output of Step~\ref{stepf:push:neighbor} or Step~\ref{stepf:delayprogress}. The condition $p \notin T$ ensures that the input requirements of the latter recursive call are satisfied.
\end{claimproof}

As the previous claim shows that the algorithm enumerates a superset of the relevant seclusion-maximal induced subgraphs, to prove \cref{thm:secluded-F-free-enum} it suffices to bound the size of the search tree generated by the algorithm, and thereby the running time and total number of induced subgraphs which are given as output. To that end, we argue that for any two successive recursive calls in the recursion tree, at least one of them makes strict progress on a relevant measure. Since no call can increase the measure, this will imply a bound on the depth of the recursion tree. Since it is easy to see that the branching factor is a constant depending on~$||\mathcal{F}||$, this will lead to the desired bound.

\begin{restatable}{claim}{claimSearchTree}
\label{claim:secluded:searchtree}
The search tree generated by the call~$\mathsf{Enum}_{\mathcal{F}}(G,S,T,k)$ has depth~$\Oh_{\ff}(k)$ and $2^{\Oh_{\ff}(k)}$~leaves.
\end{restatable}
\begin{claimproof}
Let $g(X)$ denote the number of partial forbidden graphs of $\ff$ contained in $G[X]$; note that~$g(X) \leq \sum_{F \in \ff} 2^{|V(F)|}$. For the running time analysis, we consider the progress measure $k + (k - \lambdastar_G(S,T)) + (g(V(G)) - g(S))$. We argue that the measure drops by at least one after two consecutive recursive calls to the algorithm. For most cases, the measure already drops in the first recursive call. First suppose that a recursive call is made in Step~\ref{stepf:local:branch}, then the third summand does not increase: $S$ does not change while~$g(V(G) \setminus \{u\}) \leq g(V(G))$. We have $\lambdastar_{G-u}(S,T) \geq \lambdastar_G(S,T) - 1$. Since $k$ \jjh{is} decreased by one, the measure strictly goes down.   
Next suppose a recursive call is made in Step~\ref{stepf:local:all}. Since $g(S \cup U) > g(S)$ by construction, and $\lambdastar_G(S \cup U,T) \geq \lambdastar_G(S, T)$, again the measure strictly goes down.
The fact that the measure drops for a recursive call in Step~\ref{stepf:u:a:i} follows akin to the arguments for Step~\ref{stepf:local:branch}. The same holds for Step~\ref{stepf:u:a:ii} akin to Step~\ref{stepf:local:all}. For a recursive call made in Step~\ref{stepf:u:a:iii}, we know by assumption that $\lambdastar_G(S, T \cup U) > \lambdastar_G(S,T)$. Since $k$ and $S$ remain the same, the measure strictly decreases.

The reasoning becomes more involved for a recursive call in Step~\ref{stepf:push}. For a recursive call in Step~\ref{stepf:push:neighbor}, we have~$g(S') \geq g(S)$ as $S \subseteq S'$, while $\lambdastar_{G-p}(S',T \setminus \{p\}) = \lambdastar_G(S,T) - 1$ since~$p$ belongs to a minimum left-restricted $(S,T)$-separator in~$G$, which is also a left-restricted minimum~$(S', T)$-separator. Since~$k$ goes down by one, the measure strictly decreases.

Finally, consider a recursive call made in Step~\ref{stepf:delayprogress} (so $p \notin T$). Note that $g(S' \cup \{p\}) \geq g(S)$ as $S \subseteq S'$, $k$ remains the same, and $\lambdastar_G(S' \cup \{p\},T) \geq \lambdastar_G(S,T)$. We distinguish three cases, depending on whether~$p$ is in some enrichment. 
\begin{itemize}
    \item If $\{p\} \in \mathcal{U}$, then actually $g(S' \cup \{p\}) > g(S)$ and the measure strictly drops. 
    \item If $\{p\}$ is not a tight enrichment of $S$, but $p \in U$ for some $U \in \mathcal{U}$, observe that $U \setminus \{p\}$ is disjoint from $S' \cup \{p\} \cup T$, 
    forms a tight enrichment of
    $S' \cup \{p\}$, and each connected component of $G[U \setminus \{p\}]$ is adjacent to $p \in S' \cup \{p\}$ as $G[U]$ is connected. It follows that in the next call Step~\ref{stepf:local} applies (which it reaches as we assumed the algorithm recurses twice) and again we make progress. 
    \item In the remainder we have $p \notin V(\mathcal{U})$. First consider the case that $\lambdastar_G(S' \cup \{p\},T) = \lambdastar_G(S' \cup \{p\}, T \cup V(\mathcal{U}))$. Then since $P = N_G(S')$ was a farthest $(S,T \cup V(\mathcal{U}))$-separator, by Lemma~\ref{lemma:separator:increase} we get that $\lambdastar_G(S' \cup \{p\},T) = \lambdastar_G(S' \cup \{p\}, T \cup V(\mathcal{U})) > \lambdastar_G(S, T \cup V(\mathcal{U})) = \lambdastar_G(S,T)$, and therefore the progress measure strictly drops. 
    
    In the remaining case we have $\lambdastar_G(S' \cup \{p\},T) < \lambdastar_G(S' \cup \{p\}, T \cup V(\mathcal{U}))$. Since $p \notin V(\mathcal{U})$, if the algorithm reaches Step~\ref{stepf:u} in the next iteration, the set of enrichments $\mathcal{U}$ remains the same. But then Step~\ref{stepf:u:a} applies, which makes progress in the measure as argued above. 
\end{itemize}
We have shown that the measure decreases by at least one after two consecutive recursive calls. 
The algorithm cannot proceed once the measure becomes negative because $g(S)$ cannot grow beyond $g(V(G))$ and whenever $k < 0$ or $\lambdastar_G(S,T) > k$ the algorithm immediately stops.
Since $g(V(G))$ is upper-bounded by a constant depending on $||\ff||$ and $|\ff|$, {we infer that the search tree has depth $\Oh_{\ff}(k)$.
Any tight enrichment detected in Step \ref{stepf:local} or Step \ref{stepf:u} can have at most $||\ff||$ vertices, so}
the branching factor is bounded by $||\ff||$.
Hence, the search tree has $2^{\Oh_{\ff}(k)}$ leaves as required. 
\end{claimproof}

 The previous claim implies that the number of seclusion-maximal connected $\ff$-free $k$-secluded induced subgraphs containing all of $S$ and none of $T$ is~$2^{\Oh_\ff(k)}$, since the algorithm outputs at most one subgraph per call and only does so in leaf nodes of the recursion tree. {As Claim~\ref{claim:secluded:searchtree} bounds the size of the search tree generated by the algorithm, the desired bound on the total running time follows from the claim below.}

\begin{restatable}{claim}{claimRunningTime}
A single iteration of~$\mathsf{Enum}_{\mathcal{F}}(G,S,T,k)$ can be implemented to run in time~$|\mathcal{F}| \cdot 2^{||\ff||} \cdot n^{||\mathcal{F}|| + \Oh(1)}$ and polynomial space.
\end{restatable}
\begin{claimproof}
{Within this proof, for a graph $F$ we abbreviate $|V(F)|$ to $|F|$.}
Deciding whether $\lambdastar_G(S,T) > \jjh{k}$, as required in Step~\ref{stepf:stop}, can be done in $\Oh(k(n+m))$ time by Theorem~\ref{thm:prelim:ford} and Observation~\ref{obs:prelim:restricted-translate}. Finding the connected components of $G$, and deciding if $S$ is contained in only one can be done in $\Oh(n+m)$ time. 
Deciding if $G[S]$ contains an induced subgraph isomorphic to some $F \in \ff$ can be done in $|\mathcal{F}| \cdot n^{||\ff|| + \Oh(1)}$ time. In the same running time we can decide if the connected component containing $S$ is $\ff$-free and contains nothing of $T$ as needed for Step~\ref{stepf:output}. 

For Step~\ref{stepf:local}, we proceed as follows.
For each $F \in \ff$, for each partial forbidden graph $F'$ of $F$ (which consists of some subset of the connected components of $F$), verify whether there is an induced subgraph of $G[S]$ isomorphic to $F'$ by checking all of the at most $n^{|F'|}$ ways in which it could appear and verifying in $\Oh(n^2)$ time if the right adjacencies are there. Keep track of which partial forbidden graphs are not present in $G[S]$. 
Next, for each partial forbidden graph $F'$ not appearing in $G[S]$, for each of the at most $n^{|F'|}$ induced subgraph isomorphisms {$\phi: V(F') \to V(G) \sm T$ 
we verify whether} each connected component of $U = \phi(F') \sm S$ is adjacent to a vertex of $S$. 
This brings the total time for Step~\ref{stepf:local} to $|\ff|\cdot 2^{||\ff||} \cdot n^{||\ff|| + \Oh(1)}$. 

In the same time we can compute $\mathcal{U}$ for Step~\ref{stepf:u} (this time, $G[U]$ should be connected rather than each component being adjacent to $S$). Then, deciding if $\lambdastar_G(S,T) < \lambdastar_G(S, T \cup V(\mathcal{U}))$ for Step~\ref{stepf:u:a} can be done in $\Oh(k(n+m))$ time by Theorem~\ref{thm:prelim:ford} and Observation~\ref{obs:prelim:restricted-translate} since $\lambdastar_G(S,T) \leq k$. Finding $U \in \mathcal{U}$ such that $\lambdastar_G(S,T \cup U) > \lambdastar_G(S, T)$ can be done in $\Oh(n^{||\ff||} \cdot k(n+m))$ time. If Step~\ref{stepf:u:a} does not apply, then automatically Step~\ref{stepf:push} does and so we get $\lambdastar_G(S,T) = \lambdastar_G(S, T \cup V(\mathcal{U}))$. 
Finally, computing the farthest left-restricted minimum $(S, T \cup V(\mathcal{U})$)-separator can be done in $\Oh(k(n+m))$ time by Lemma~\ref{lem:closestfarthestsep} and Observation~\ref{obs:prelim:restricted-translate}. It is easy to see the steps above can be carried out using polynomial space.
\end{claimproof}

This concludes the proof of \cref{thm:secluded-F-free-enum}.
\end{proof}

\section{Applications}\label{sec:applications} 

As applications of Theorem~\ref{thm:secluded-F-free-enum}, we derive faster algorithms for two problems studied in the literature. The first problem is formally defined as follows~\cite{GolovachHLM20} for any finite set~$\mathcal{F}$ of undirected graphs.

\defparproblem{\textsc{Connected $k$-secluded $\mathcal{F}$-free subgraph}}
{Graph~$G$, integer~$k$, weight function $w \colon V(G) \to \mathbb{Z}_{> 0}$}
{$k$}
{Find a connected $k$-secluded set~$C \subseteq V(G)$ for which~$G[C]$ is $\mathcal{F}$-free which maximizes~$\sum_{v \in C} w(v)$.}

A single-exponential algorithm for this problem follows easily from~\cref{cor:secluded-F-free-enum-all}.

\corConnectedFFree*
\begin{proof}
Since the weights are positive, any maximum-weight solution to the problem is seclusion-maximal with respect to being $k$-secluded, connected, and $\mathcal{F}$-free. We can therefore solve an instance~$(G,\jjh{k},w)$ as follows. 
\mic{Invoke \cref{cor:secluded-F-free-enum-all} to} enumerate a superset of the all seclusion-maximal connected $\mathcal{F}$-free $k$-secluded induced subgraphs containing~$S := \{v\}$. For each enumerated set~$C$, check whether it is indeed~$\mathcal{F}$-free in time~$n^{||\mathcal{F}||+\Oh(1)}$. The heaviest, taken over all choices of~$v$ and~$C$, is given as the output.
\end{proof}
\newpage

Our second application concerns deletion problems to scattered graph classes, which are defined for finite sequences~$(\Pi_1, \ldots, \Pi_d)$ of graph classes.

\defparquestion{\textsc{$(\Pi_1, \ldots, \Pi_d)$-deletion}}
{Graph~$G$ and integer~$k$.}
{$k$}
{Is there a vertex set~$X \subseteq V(G)$ of size at most~$k$, such that for each connected component~$C$ of~$G-X$ there exists~$i \in [d]$ such that~$C \in \Pi_i$?}

By exploiting the fact that each connected component of~$G-X$ is $k$-secluded, we can obtain single-exponential FPT algorithms for this problem when each graph class~$\Pi$ is characterized by a finite number of forbidden induced subgraphs. In the following statement, both $\Oh_\Pi$'s hide factors depending on the choice of~$(\Pi_1, \ldots, \Pi_d)$.

\thmInducedScattered*
\begin{proof}
We describe an algorithm for the problem. If $k < 0$, report that it is a no-instance. If there is a connected component that belongs to $\Pi_i$ for some $i \in [d]$, then delete the component and continue (\cite[Reduction Rule 1]{JacobKMR23}). If the graph becomes empty, return that it is a yes-instance. Otherwise, if $k = 0$, report that it is a no-instance. 

In the remainder we have $k > 0$ and $G$ non-empty. Pick a vertex $v \in V(G)$. There are two cases; $v$ either belongs to the solution set $X$, or belongs to a component in $G-X$ in some graph class $\Pi_i$. We perform branching to cover both options. For the first option, recursively call the algorithm on $G-v$ searching for a solution of size $k-1$. For the second option, for each $i \in [d]$ and $s \in [k]$, apply Theorem~\ref{thm:secluded-F-free-enum} to enumerate (a superset of) the seclusion-maximal connected $\mathcal{F}_i$-free $s$-secluded subgraphs containing $v$. Note that the theorem implies this output has at most $c^s$ elements for some constant $c$.
For each of the enumerated subgraphs $C$ such that $G[C] \in \Pi_i$ and $|N_G(C)| = s$, recursively call the algorithm on $G-N_G[C]$ searching for a solution of size $k - |N_G(C)|$. Output yes if and only if one of the recursive calls results in a yes-instance.

For correctness of the algorithm, we argue that the enumeration of seclusion-maximal secluded subgraphs suffices. Suppose there is a solution $X$ not containing $v$ such that the component $C$ containing $v$ in $G-X$ belongs to $\Pi_i$. If~$C$ was among the output of the enumeration algorithm, it is easy to see the algorithm is correct. Suppose that~$C$ was not enumerated because it is not seclusion-maximal. For this choice of $i$ and $s = |N_G(C)|$, the enumeration included some connected $\mathcal{F}_i$-free $s$-secluded subgraph $C'$ with $C \subseteq C'$ and~$|N_G(C')| \leq |N_G(C)|$. Since the target graph classes are hereditary and graph $G - N_G[C]$ admits solution $X \sm N_G(C)$ of size at most $k - |N_G(C)|$, then its induced subgraph $G - N_G[C']$ admits a solution $X'$ of size at most $k - |N_G(C')|$.
Hence, $X' \cup N_G(C')$ is also a valid solution for $G$ of size at most $k$.
We conclude that the branching algorithm always finds a solution if there is one. 

We turn to the running time. Let $T(k)$ denote the number of leaves in the recursion tree for a call with parameter $k$, where $T(0) = 1$. By grouping the secluded subgraphs by their neighborhood size, observe that this satisfies $T(k) = T(k-1) + d\cdot \sum_{i=1}^k c^i \cdot T(k-i) \leq (d+1) \cdot \sum_{i=1}^k c^i \cdot T(k-i)$ (the inequality clearly holds if $c \geq 1$). By induction we argue that $T(k) \leq ((d+1)2c)^k$, which trivially holds if $k=0$. Suppose that it holds for all values below~$k$; then we derive:

\begin{align*}
    T(k) &\le (d+1) \cdot \sum_{i=1}^k c^i \cdot T(k-i) & \mbox{By grouping on neighborhood size.} \\
    &\le (d+1) \cdot \sum_{i=1}^k c^i \cdot \left ((d+1)2c \right)^{k-i} & \mbox{By induction.} \\
    &\leq \left ((d+1)c \right)^k \cdot \sum_{i=1}^k 2^{k-i} & \mbox{Using~$(d+1)^{k-i}\leq (d+1)^{k-1}$.} \\
    &\leq ((d+1)2c)^k. & \mbox{Since~$\sum_{i=0}^{k-1} 2^i < 2^k$.}
\end{align*}

Since the depth of the recursion tree is at most $k$, the recursion tree has at most $k \cdot ((d+1)2c)^k$ nodes. Finally we consider the running time per node of the recursion tree. Finding the connected components can be done in $\Oh(n+m)$ time. Checking if one of them belongs to $\Pi_i$ for some $i \in [d]$ can be done in $n^{\Oh_\Pi(1)}$ time.
The time needed for the $d \cdot k$ calls to Theorem~\ref{thm:secluded-F-free-enum} is $dk \cdot 2^{\Oh_\Pi(k)}\cdot n^{\Oh_\Pi(1)}$. Since $d$ and $c$ are constants, we get the claimed running~time.

Note that since Theorem~\ref{thm:secluded-F-free-enum} uses polynomial space, and we process its output one at a~time without storing it, we conclude that the described algorithm uses polynomial space.
\end{proof}

\section{Conclusion} \label{sec:conclusion}

{We have introduced a new algorithmic primitive based on secluded connected subgraphs which generalizes important separators.}
The high-level idea behind {the algorithm} is \emph{enumeration via separation}: by introducing an artificial set~$T$ and considering the more general problem of enumerating secluded subgraphs containing~$S$ but disjoint from~$T$, we can analyze the progress of the recursion in terms of the size of a minimum (left-restricted)~$(S,T)$-separator. We expect this idea to be useful in scenarios beyond the one studied here.

We presented a single-exponential, polynomial-space FPT \jjh{algorithm} to enumerate the family of seclusion-maximal connected $\mathcal{F}$-free subgraphs for finite~$\mathcal{F}$, making it potentially viable for practical use~\cite{PilipczukZ18}. 
The combination of single-exponential running time and polynomial space usage sets our approach apart from others such as recursive understanding~\cite{ChitnisCHPP16,CyganKLPPSW21,LokshtanovR0Z18} and treewidth reduction~\cite{MarxR14}. Algorithms exploiting half-integrality of the linear-programming relaxation or other discrete relaxations also have these desirable properties, though~\cite{CyganPPW13,Guillemot11a,IwataWY16,IwataYY18,Xiao10}. Using this approach, Iwata, Yamaguchi, and Yoshida~\cite{IwataYY18} even obtained a \emph{linear-time} algorithm in terms of the number of vertices~$n$, solving (vertex) \textsc{Multiway Cut} in time~$2^k \cdot k \cdot (n + m)$. At a high level, there is some resemblance between their approach and ours. They work on a discrete relaxation of deletion problems in graphs which are not standard LP-relaxations, but are based on relaxations of a \emph{rooted} problem in which only constraints involving a prescribed set~$S$ are active. This is reminiscent of the fact that we enumerate secluded subgraphs containing a prescribed set~$S$. Their branching algorithms are based on the notion of an extremal optimal solution to the LP relaxation, which resembles our use of the farthest minimum left-restricted~$(S,T)$-separator. However, the two approaches diverge there. To handle problems via their approach, they should be expressible as a 0/1/ALL CSP. Problems for which the validity of a solution can be verified by unit propagation (such as \textsc{Node Unique Label Cover}, \textsc{Node Multiway Cut}, \textsc{Subset} and \textsc{Group Feedback Vertex Set}) belong to this category, but it seems impossible to express the property of being~$\mathcal{F}$-free for arbitrary finite sets~$\mathcal{F}$ in this framework.

\bmp{The branching steps underlying our algorithm were informed by the structure of the subgraphs induced by certain vertex sets. In the considered setting, where certain possibly disconnected structures are not allowed to appear inside~$C$, it is necessary to characterize the forbidden sets in terms of the graph structure they induce. But when the forbidden sets are connected, we believe our proof technique can be used in a more general setting to establish the following. For any~$n$-vertex graph~$G$, non-empty vertex set~$S \subseteq V(G)$, potentially empty~$T \subseteq V(G) \setminus S$, integer~$k$, and collection~$F_1, \ldots, F_m \subseteq V(G)$ of vertex sets of size at most~$\ell$ which are connected in~$G$, the number of $k$-secluded induced subgraphs~$G[C]$ which are seclusion-maximal with respect to being connected, not containing any set~$F_i$, and satisfying~$S \subseteq C \subseteq V(G) \setminus T$, is bounded by~$(2+\ell)^{\Oh(k)}$, and a superset of them can be enumerated in time~$(2+\ell)^{\Oh(k)} \cdot m \cdot n^{\Oh(1)}$ and polynomial space. 
\mic{The reason why dealing with general connected obstacles is feasible is that whenever $F_i \cap C \ne \emptyset$ then also $F_i \cap N(C) \ne \emptyset$; this allows us to always make progress using the simpler branching strategy without keeping track of partial forbidden graphs.}
The corresponding generalization for \emph{disconnected} vertex sets~$F_i$ is false, even for~$|F_i| = 2$. To see this, consider a graph consisting of a cycle on~$2m+1$ vertices consecutively labeled~$s, a_1, \ldots, a_m, b_1, \ldots, b_m$ with~$F_i = \{a_i, b_i\}$ for each~$i \in [m]$, in which the number of relevant seclusion-maximal 2-secluded sets containing~$s$ is~$\Omega(m)$.}

We leave it to future work to consider generalizations of our ideas to \emph{directed graphs}. Since important separators also apply in that setting, we expect the branching step in terms of left-restricted minimum separators to be applicable in directed graphs as well. However, there are multiple ways to generalize the notion of a connected secluded induced subgraph to the directed setting: one can consider weak connectivity, strong connectivity, or a rooted variant where we consider all vertices reachable from a source vertex~$x$. Similarly, one can define seclusion in terms of the number of in-neighbors, out-neighbors, or both.

\bibliography{bib}

\end{document}